\documentclass[11pt]{article}

\usepackage{amsmath}
\usepackage{amsthm}
\usepackage{amssymb}
\usepackage{dsfont}
\usepackage{enumerate}
\usepackage{mathrsfs}
\usepackage[pdftex,left=1in,top=1in,bottom=1in,right=1in]{geometry}
\usepackage{graphicx}
\usepackage{caption}
\usepackage{bm}
\usepackage{color}
\usepackage{commath}
\usepackage{hyperref}
\usepackage{indentfirst}

\newcommand{\pflip}{p_{\mathsf{flip}}}
\newcommand{\C}{\mathds{C}}

\newcommand{\R}{\mathbb{R}}
\newcommand{\E}{\mathds{E}}

\newcommand{\cD}{\mathcal{D}}

\newcommand{\cS}{\mathcal{S}}
\newcommand{\cT}{\mathcal{T}}

\newcommand{\eps}{\epsilon}

\newcommand{\prn}[1]{\left( #1 \right)}

\newtheorem{thm}{Theorem}

\newtheorem{lem}[thm]{Lemma}

\newtheorem{coro}[thm]{Corollary}

\newcommand{\Ch}{\mathsf{Ch}}

\newcommand{\cR}{\mathcal{R}}

\newcommand{\poly}{\textrm{poly}}

\newcommand{\obits}{\left\{ -1,1\right\}}

\newcommand{\Rep}{\mathsf{Rep}}
\newcommand{\Flip}{\mathsf{Flip}}

\newcommand{\RIchannel}{\Ch_{(M,\cR, \pflip)}} 
\newcommand{\channel}{\Ch_{M}} 

\newcommand{\brk}[1]{\left[ #1 \right]}

\definecolor{cerulean}{rgb}{0.0, 0.48, 0.65}

\let\originalleft\left
\let\originalright\right
\renewcommand{\left}{\mathopen{}\mathclose\bgroup\originalleft}
\renewcommand{\right}{\aftergroup\egroup\originalright}

\title{Mean-Based Trace Reconstruction over\\ Oblivious Synchronization Channels
}

\author{
Mahdi Cheraghchi\thanks{University of Michigan --  Ann Arbor. Email: \texttt{mahdich@umich.edu}}\and Joseph Downs\thanks{University of Michigan --  Ann Arbor. Email: \texttt{josdowns@umich.edu}}\and João Ribeiro\thanks{Carnegie Mellon University. Part of the work was done while at Imperial College London. Email: \texttt{jlourenc@andrew.cmu.edu}}\and Alexandra Veliche\thanks{University of Michigan --  Ann Arbor. Email: \texttt{aveliche@umich.edu}}

\thanks{This material is based upon work supported by the National Science Foundation under Grant No.\ CCF-2006455. A preliminary version of this work was presented at the 2021 IEEE International Symposium on Information Theory~\cite{CDRV21}.}
}

\date{}

\begin{document}

\maketitle

\begin{abstract}
Mean-based reconstruction is a fundamental, natural approach to worst-case trace reconstruction over channels with synchronization errors.
It is known that $\exp(\Theta(n^{1/3}))$ traces are necessary and sufficient for mean-based worst-case trace reconstruction over the deletion channel, and this result was also extended to certain channels combining deletions and geometric insertions of uniformly random bits.
In this work, we use a simple extension of the original complex-analytic approach to show that these results are examples of a much more general phenomenon.
We introduce \emph{oblivious synchronization channels}, which map each input bit to an arbitrarily distributed sequence of replications and insertions of random bits.
This general class captures all previously considered synchronization channels.
We show that for any oblivious synchronization channel whose output length follows a sub-exponential distribution either mean-based trace reconstruction is impossible or $\exp(O(n^{1/3}))$ traces suffice for this task.
\end{abstract}

\section{Introduction}\label{sec:intro}
When any length-$n$ message $x\in\obits^n$ is sent through a noisy channel $\Ch$, the channel modifies the input $x$ in some way to produce a distorted copy of $x$, which we call a \emph{trace}.
The goal of worst-case trace reconstruction over $\Ch$ is to design an algorithm which recovers any input string $x\in\obits^n$ with high probability from as few independent and identically distributed (i.i.d.)\ traces as possible.
This problem was first introduced by Levenshtein~\cite{Lev01,Lev01b}, who studied it over combinatorial channels causing synchronization errors, such as worst-case deletions and insertions of symbols and certain discrete memoryless channels.
Trace reconstruction over the \emph{deletion channel}, which independently deletes each input symbol with some probability, was first considered by Batu, Kannan, Khanna, and McGregor~\cite{BKKM04}.
Some of their results were quickly generalized to what we call the \emph{geometric insertion-deletion channel}~\cite{KM05,VS08}, which prepends a geometric number of independent, uniformly random symbols to each input symbol and then deletes it with a given probability.
Both the deletion and geometric insertion-deletion channels are examples of discrete memoryless synchronization channels~\cite{Dob67,CR20}.

Holenstein, Mitzenmacher, Panigrahy, and Wieder~\cite{HMPW08} were the first to obtain non-trivial worst-case trace reconstruction algorithms for the deletion channel with constant deletion probability.
They showed that $\exp(\widetilde{O}(\sqrt{n}))$ traces suffice for \emph{mean-based} reconstruction of any input string with high probability, where $\widetilde{O}(\cdot)$ hides polylogarithmic factors.
By mean-based reconstruction, we mean that the reconstruction algorithm only requires knowledge of the expected value of each trace coordinate.
In general, this procedure works as follows:
Let $Y_x=(Y_{x,1},Y_{x,2},\dots)$ denote the \emph{trace distribution} on input $x\in\{-1,1\}^n$ and $Y'_x$ denote the infinite string obtained by padding $Y_x$ with zeros on the right.
The \emph{mean trace} $\mu_x$ is given by
\begin{equation*}
    \mu_x=(\E[Y'_{x,1}],\E[Y'_{x,2}],\dots).
\end{equation*}
As the first step, the algorithm estimates $\mu_x$ from $t$ traces $T^{(1)},T^{(2)},\dots,T^{(t)}$ sampled i.i.d.\ according to $Y_x'$ via the empirical means
\begin{equation}\label{eq:muhat}
    \widehat{\mu}_i=\frac{1}{t}\sum_{j=1}^t T^{(j)}_i, \quad i=1,2,\dots.
\end{equation}
Subsequently, it outputs the string $\widehat{x}\in\obits^n$ that minimizes $\|\mu_{\widehat{x}}-\widehat{\mu}\|_1$.
If $t=t(n)$ is large enough, we have $\widehat{x}=x$ with high probability over the randomness of the traces.
Because of their structure, pinpointing the number of traces required for mean-based reconstruction over any channel reduces to bounding $\|\mu_x-\mu_{x'}\|_1$ for any pair of distinct strings $x,x'\in\obits^n$.
Overall, mean-based reconstruction is a natural paradigm, and it is not only useful over channels with synchronization errors.
For example, $O(\log n)$ traces suffice for mean-based reconstruction over the binary symmetric channel, which is optimal.

More recently, an elegant complex-analytic approach was employed concurrently by De, O'Donnell, and Servedio~\cite{DOS17} and by Nazarov and Peres~\cite{NP17} to show that \mbox{$\exp(O(n^{1/3}))$} traces suffice for mean-based worst-case trace reconstruction not only over the deletion channel with constant deletion probability, but also over the more general geometric insertion-deletion channel we described previously.\footnote{Nazarov and Peres~\cite{NP17} consider a slightly modified geometric insertion-deletion channel: First, a geometric number of independent, uniformly random symbols is added independently before each input symbol. Then, the resulting string is sent through a deletion channel. The analysis is similar to that of the geometric-insertion channel.}
Remarkably, $\exp(\Omega(n^{1/3}))$ traces were shown to also be necessary for mean-based reconstruction over the deletion channel.

Given the fundamental nature of mean-based reconstruction and this state of affairs, the following question arises naturally: \emph{Are these results examples of a much more general phenomenon?} In particular, is it true that $\exp(O(n^{1/3}))$ traces suffice for mean-based trace reconstruction over a much more general class of synchronization channels? In this work, we introduce and study the general class of \emph{oblivious synchronization channels}. We use the term \emph{oblivious} to describe channels that behave in an i.i.d.\ manner for each input bit and use randomness that is independent of the input.
This class is a significant generalization of all synchronization channels previously studied in the context of trace reconstruction.
We make progress in this direction by showing that a simple extension of the analysis from~\cite{DOS17,NP17} yields the same result for all oblivious synchronization channels satisfying a mild assumption already present in~\cite{DOS17,NP17}.

Research in this direction has other practical and theoretical implications.
First, studying trace reconstruction over channels introducing more complex synchronization errors than simple i.i.d.\ deletions is fundamental for the design of reliable DNA-based data storage systems with nanopore-based sequencing~\cite{YKG+15,YGM17,OAC+18}.
Second, understanding the structure of the mean trace of a string is a natural information-theoretic problem which may lead to improved capacity bounds and coding techniques for channels with synchronization errors, both notoriously difficult problems (see the extensive surveys~\cite{Mit09,MBT10,CR20,HS21}).

\subsection{Related Work}

Besides the works mentioned above, there has been significant recent interest in various notions of trace reconstruction.
The mean-based approach of~\cite{DOS17,NP17} has proven useful to some problems incomparable to our general setting: the deletion channel with position- and symbol-dependent deletion probabilities satisfying strong monotonicity and periodicity assumptions~\cite{HHP17}; a combination of the geometric insertion-deletion channel and random shifts of the output string as an intermediate step in the design of \emph{average-case} trace reconstruction algorithms (which are only required to succeed with high probability when the input is uniformly random)~\cite{PZ17,HPP18}; trace reconstruction of trees with i.i.d.\ deletions of vertices~\cite{DRR19}; trace reconstruction of matrices with i.i.d.\ deletions of rows and columns~\cite{KMMP19}; trace reconstruction of circular strings over the deletion channel~\cite{NR21}.
In another direction, Grigorescu, Sudan, and Zhu~\cite{GSZ21} and Sima and Bruck~\cite{SB21} studied the performance of mean-based reconstruction for distinguishing between strings at low Hamming or edit distance from each other over the deletion channel.

Different complex-analytic methods have been used to obtain the current best upper bound of $\exp(\widetilde{O}(n^{1/5}))$ traces on the trace complexity of the deletion channel~\cite{C21}, as well as upper bounds for trace reconstruction of ``smoothed'' worst-case strings over the deletion channel~\cite{CDLSS20}.
However, mean-based reconstruction remains the state-of-the-art approach for the geometric insertion-deletion channel.

Other related problems considered include the already-mentioned average-case trace reconstruction problem over the deletion and geometric insertion-deletion channels~\cite{BKKM04,HMPW08,MPV14,PZ17,HPP18}, trace reconstruction over the deletion and geometric insertion-deletion channels with vanishing deletion probabilities~\cite{BKKM04,KM05,VS08,CDLSS20b}, trace complexity lower bounds for the deletion channel~\cite{BKKM04,MPV14,HL20,Cha20}, trace reconstruction of coded strings over the deletion channel~\cite{CGMR20,BLS19}, approximate trace reconstruction~\cite{DRRS20,CDK21,CP21,CDLSS21}, alternative trace reconstruction models motivated by immunology~\cite{BPRS20}, and population recovery over the deletion and geometric insertion-deletion channels~\cite{BCFSS19,BCSS19,Nar20}.

\subsection{Notation}
For convenience, we denote discrete random variables and their corresponding distributions by uppercase letters, such as $X$, $Y$, and $Z$.
The expected value of $X$ is denoted by $\E[X]$.
Sets are denoted by calligraphic uppercase letters such as $\cS$ and $\cT$, and we write $[n]:=\{1,2,\dots,n\}$.
The open disk of radius $r$ centered at $z\in\C$ is $\cD_r(z):=\{z'\in\C:\abs{z-z'}<r\}$.
The $1$-norm of vector $x$ is denoted by $\|x\|_1$.
The concatenation of strings $x$ and $y$ is denoted by $x\|y$.
For any random variable $X$ supported on the set of non-negative integers, we denote its probability generating function by $g_X(\cdot)$.
For two functions $f$ and $g$, we use $f(x)\sim g(x)$ to mean that $\lim_{x\to\infty}\frac{f(x)}{g(x)}=1$.

\subsection{Channel Model}\label{sec:channel}

We introduce and study a general model of discrete memoryless synchronization channels that, in particular, captures the models studied in~\cite{KM05,VS08,DOS17,NP17,CDRV21}.
An \emph{oblivious synchronization channel} $\channel$ is characterized by a random variable $M$ and a corresponding collection of randomized functions \mbox{$F_M:\obits\to\obits^M$}.
To avoid trivial settings where trace reconstruction is impossible, we require that $\Pr[M>0]>0$.
On each input $x_i\in\obits$, the channel samples $m$ from $M$
and decides which positions of the output are $x_i$ replicated, flipped, have a value of -1, or have a value of 1.
The sets corresponding to these positions are denoted by $\Rep$, $\Flip$, $C_+$, and $C_-$, respectively.
These sets are so named because $\Rep$ replicates the input bit, $\Flip$ flips it, $C_+$ is constantly $+1$, and $C_-$ is constantly $-1$.
Note that these sets partition the output length $[m]$ according to an arbitrary distribution independent of the input $x$.
This sampling determines a function $f:\obits\to\obits^m$.
In other words,
\begin{align*}
    \Rep&:=\{j\in[m] : f(-1)_j=-1\text{ and } f(1)_j=1\},\\
    \Flip&:=\{j\in[m] : f(-1)_j=1\text{ and } f(1)_j=-1\},\\
    C_+&:=\{j\in[m] : f(-1)_j=1=f(1)_j\},\\
    C_-&:=\{j\in[m] : f(-1)_j=-1=f(1)_j\},
\end{align*}
where $f(y)_j$ denotes the $j^\text{th}$ coordinate of $f(y)$.
The channel evaluates this function at $x_i$ to obtain $f(x_i)$.
From here onward, we use $\Rep$ and $\Flip$ to denote the sets as well as their corresponding distributions.

As an example, we now show how the replication-insertion channel from~\cite{CDRV21} is an instance of the oblivious synchronization channel. Note that the deletion and geometric insertion-deletion channels are particular examples of the replication-insertion channel.
A replication-insertion channel $\RIchannel$ is determined by three parameters a flip probability $\pflip\in[0,\frac{1}{2})$, a sub-exponential output length distribution $\mathcal{M}$ over the non-negative integers, and a replication distribution $\mathcal{R}$ over subsets of $[M]$.
For any given bit $x\in\obits$, the channel samples $(M_i,R_i)$ according to the joint distribution $(M,R)$.
It then produces the output string $Y_x\in\obits^{M_i}$ bit-wise by defining $Y_{x,j}=-x$ with probability $\pflip$ and $Y_{x,j}=x$ with probability $1-\pflip$ if $j\in R_i$ and sample $Y_{x,j}$ uniformly from $\obits$  otherwise.
Let $m=M_i$ and $F_m$ be the collection of randomized functions, where the randomness comes from sampling $\pflip\leftarrow [0,\frac{1}{2})$ and $R_i\leftarrow 2^{[m]}$.
We sample the sets $\Rep, \Flip, C_+,$ and $C_-$ in the following way:
For each position $j\in[m]$, if $j\in R_i$ we put $j$ in $\Rep$ with probability $1-\pflip$ and in $\Flip$ with probability $\pflip$.
Otherwise, we choose a bit uniformly at random.
If the random bit is $1$, we put $j$ in $C_+$, and if it is $-1$, we put $j$ in $C_-$.

\subsection{Our Contributions}

Our main theorem shows that previous results on mean-based trace reconstruction over the deletion and geometric insertion-deletion channels are examples of a much more general phenomenon.

\begin{thm}\label{thm:mainresult}

    Let $\channel$ be an oblivious synchronization channel where $M$ is a sub-exponential random variable.\footnote{A random variable $M$ is \emph{sub-exponential} if there exists a constant $\alpha>0$ such that $\Pr[|M|\geq \tau]\leq 2 e^{-\alpha \tau}$ for all $\tau\geq 0$.}
    Define the random variables $W_R$ and $W_F$ with probability mass functions
    \begin{equation*}
        W_R(j):=\frac{\Pr[j+1\in\Rep]}{\E[|\Rep|]} \text{ and }
        W_F(j):=\frac{\Pr[j+1\in\Flip]}{\E[|\Flip|]}, j=0,1,2,...,
    \end{equation*}
    and let $g_{W_R}$, $g_{W_F}$ be their respective probability generating functions. If $\Flip$ (or $\Rep$) is always the empty set, then define $g_{W_F}=0$ (or $g_{W_R}=0$). If
    \begin{equation}\label{eq:mbtrcond}
        \E[|\Rep|] \cdot g_{W_R}(\cdot) \not\equiv \E[|\Flip|] \cdot g_{W_F}(\cdot),
    \end{equation}
    then $\exp(O(n^{1/3}))$ traces are sufficient for mean-based trace reconstruction over $\channel$ with success probability \mbox{$1-e^{-\Omega(n)}$}.
    If~\eqref{eq:mbtrcond} is not satisfied, then mean-based trace reconstruction is impossible.
\end{thm}

Note that many common distributions are sub-exponential
, including geometric, Poisson, and all finitely-supported distributions.
In general, $\E[|\Rep|]$ or $\E[|\Flip|]$ could be infinite, so these distributions are not always well defined.
However, this is not a problem because our theorem only applies to channels where $M$ is a sub-exponential  random variable.
Because $M$ is sub-exponential, it has finite expectation.
Here $|\Rep|$ and $|\Flip|$ are both upper bounded by $M$, so they also have finite expectation and $W_R, W_F$ are valid distributions.

\section{Proof of Theorem~\ref{thm:mainresult}}\label{sec:extensioninsertdelete}

Fix an oblivious synchronization channel $\channel$, where $M$ is a sub-exponential random variable and
$\E[|\Rep|] \cdot g_{W_R}(z) \neq \E[|\Flip|] \cdot g_{W_F}(z)$ for some $z\in\C$.
To every string \mbox{$x\in\{-1,1\}^n$}, we can associate a polynomial $P_x$ over $\C$ defined as
\begin{equation*}
    P_x(z):=\sum_{i=1}^n x_i z^{i-1}.
\end{equation*}
Then, using the definition of mean trace above, we define the mean trace power series $\overline{P}_x$ as
\begin{equation*}
    \overline{P}_x(z):=\sum_{i=1}^{\infty} \mu_{x,i} z^{i-1},
\end{equation*}
where $\mu_{x,i}$ denotes the $i^\text{th}$ coordinate of the mean trace.
Let $N>0$ and denote the mean trace truncated at the $N^\text{th}$ coordinate by
\begin{equation*}
    \mu^N_x:=(\mu_{x,1},\dots,\mu_{x,N}).
\end{equation*}

To prove Theorem~\ref{thm:mainresult}, we will show that there exists a constant $C>0$ such that for a large enough $n$, appropriate $N$, and \emph{any} distinct input strings $x,x'\in\obits^n$, their truncated mean traces satisfy
\begin{equation}\label{eq:truncdifferencebound}
    \norm{\mu_x^N-\mu_{x'}^N}_1=\sum_{i=1}^N\abs{\mu_{x,i}-\mu_{x',i}}\geq \delta(n):= e^{-C n^{1/3}}.
\end{equation}
This implies that $\exp(O(n^{1/3}))$ traces suffice for mean-based worst-case trace reconstruction as follows:
Let $x$ be the true input and suppose that we have access to \mbox{$t:=n/\delta(n)^2=\exp(O(n^{1/3}))$} traces.
Then a direct application of the Chernoff bound and a union bound over all coordinates $i=1,\dots,N$ shows that the empirical mean trace $\widehat{\mu}^N=(\widehat{\mu}_1,\dots,\widehat{\mu}_N)$ defined in~\eqref{eq:muhat} satisfies
\begin{equation}\label{eq:goodestimate}
    \norm{\widehat{\mu}^N-\mu_x^N}_1\leq \frac{\delta(n)}{4}
\end{equation}
with probability at least $1-e^{-\Omega(n)}$ over the randomness of the traces.
On the other hand, if~\eqref{eq:goodestimate} holds, we can combine it with~\eqref{eq:truncdifferencebound} and the triangle inequality to get
\begin{equation*}
    \norm{\widehat{\mu}^N-\mu_{x'}^N}_1\geq \frac{3\delta(n)}{4}
\end{equation*} for all $x'\neq x$.
This allows us to recover $x$ naively from $\widehat{\mu}$ by computing $\mu^N_{\widehat{x}}$ for every $\widehat{x}\in\obits^n$ and outputting the $\widehat{x}$ that minimizes $\norm{\widehat{\mu}^N-\mu^N_{\widehat{x}}}$.

We prove~\eqref{eq:truncdifferencebound} by relating $\norm{\widehat{\mu}^N-\mu_x^N}_1$ to $|\overline{P}_x(z)-\overline{P}_{x'}(z)|$ for an appropriate choice of $z\in\C$.
Assuming that $|z|\geq 1$, by the triangle inequality we have
\begin{align*}
    \abs{\overline{P}_x(z)-\overline{P}_{x'}(z)}\nonumber&\leq\sum_{i=1}^\infty\abs{\mu_{x,i}-\mu_{x',i}}|z|^{i-1}\nonumber\\
    &=\sum_{i=1}^N \abs{\mu_{x,i}-\mu_{x',i}}|z|^{i-1}+\sum_{i=N+1}^\infty \abs{\mu_{x,i}-\mu_{x',i}}|z|^{i-1}\nonumber\\
    &\leq |z|^{N}\left\|\mu^N_x-\mu^N_{x'}\right\|_1+\sum_{i=N+1}^\infty \abs{\mu_{x,i}-\mu_{x',i}}|z|^{i-1}
\end{align*}
for every $z\in\C$ such that $|z|\geq 1$.
Rearranging, it follows that $\left\|\mu^N_x-\mu^N_{x'}\right\|_1$ is lower-bounded by
\begin{equation}\label{eq:meantopow}
    |z|^{-N}\left(\abs{\overline{P}_x(z)-\overline{P}_{x'}(z)}-\sum_{i=N+1}^\infty \abs{\mu_{x,i}-\mu_{x',i}}|z|^{i-1}\right)
\end{equation}
for any such $z$.
The lower bound in~\eqref{eq:truncdifferencebound}, and thus Theorem~\ref{thm:mainresult}, follows by combining~\eqref{eq:meantopow} with the next two lemmas, each bounding a different term in the right-hand side of~\eqref{eq:meantopow}.
\begin{lem}\label{lem:pow}
    There exist constants $c_1,c_2>0$ such that for $n$ large enough and any distinct strings $x,x'\in\obits^n$, it holds that $\abs{\overline{P}_x(z)-\overline{P}_{x'}(z)}\geq e^{-c_1 n^{1/3}}$ for some $z$ satisfying $1\leq |z|\leq e^{c_2 n^{-2/3}}$.
\end{lem}

\begin{lem}\label{lem:tail}
    If there exists a constant $c_3>0$ such that $1\leq |z|\leq e^{c_3 n^{-2/3}}$, then there exist constants $c_4,c_5>0$ such that $N=c_4 n$ implies
    \begin{equation*}
        \sum_{i=N+1}^\infty \abs{\mu_{x,i}-\mu_{x',i}}|z|^{i-1}\leq e^{-c_5 n}
    \end{equation*}
    for all distinct $x,x'\in\obits^n$ when $n$ is large enough.
\end{lem}
Invoking Lemmas~\ref{lem:pow} and~\ref{lem:tail}, we have that for $n$ large enough and any distinct $x,x'\in\obits^n$, there exists an appropriate choice $z^\star\in\C$ possibly depending on $x$ and $x'$ which satisfies \mbox{$1\leq |z^\star|\leq e^{c_2 n^{-2/3}}$} and by setting $z=z^\star$ and $N=c_4 n$ in~\eqref{eq:meantopow} yields
\begin{align*}
\norm{\mu_x^N-\mu_{x'}^N}_1 &\geq e^{-c_4\cdot c_2 n^{1/3}}\left(e^{-c_1 n^{1/3}}-e^{-c_5 n}\right)\\
&\geq e^{-C n^{1/3}}
\end{align*}
for some constant $C>0$, implying~\eqref{eq:truncdifferencebound}.

We prove Lemmas~\ref{lem:pow} and~\ref{lem:tail} in Sections~\ref{sec:proofpow} and~\ref{sec:prooftail}, respectively, which completes the argument.

\section{Proof of Lemma~\ref{lem:pow}}\label{sec:proofpow}

Our proof of Lemma~\ref{lem:pow} follows the blueprint of~\cite[Sections 4 and 5]{DOS17} and~\cite[Sections 2 and 3]{NP17}.
The key differences lie in Lemmas~\ref{lem:choicez} and~\ref{lem:changevar} below.
Lemma~\ref{lem:choicez} requires analyzing the local behavior of the inverse of an \emph{arbitrary} probability generating function (PGF) in the complex plane around $z=1$.
Remarkably, the desired behavior follows by combining the standard inverse function theorem for analytic functions with basic properties of PGFs.
In contrast, the PGFs associated to the deletion and geometric insertion-deletion channels treated in~\cite{DOS17,NP17,PZ17,HPP18} are all M\"obius transformations, meaning that their inverses have simple explicit expressions which were then easily analyzed directly.
Lemma~\ref{lem:changevar} generalizes~\cite[Section 4 and Appendix A.3]{DOS17} and~\cite[Lemmas 2.1 and 5.2]{NP17} to arbitrary sub-exponential oblivious synchronization channels well beyond the deletion and geometric insertion-deletion channels.

As a first step, we show that the mean trace power series $\overline{P}_x$ is related to the input polynomial $P_x$ through a change of variable.
This allows us to bound $\abs{\overline{P}_x(z)-\overline{P}_{x'}(z)}$ in terms of $\abs{P_x(w)-P_{x'}(w)}$ for some $w$ related to $z$.
To do this, we first derive an expression for $\overline{P}_x(z)$.





\begin{lem}\label{lem:meantracediff}
Suppose $\E[M]>0$ is finite. Let $W_R,W_F$ be distributions corresponding to $\Rep$ and $\Flip$, respectively, with associated probability mass functions
\begin{equation*}
    W_R(j):=\frac{\Pr[j+1\in\Rep]}{\E[|\Rep|]} \text{ , }
    W_F(j):=\frac{\Pr[j+1\in\Flip]}{\E[|\Flip|]}, j=0,1,2,...,.
\end{equation*}
Also, let $g_{W_R},g_{W_F}$, and $g_M$ be the probability generating functions corresponding to $W_R, W_F,$ and $M$.
If either $\E[|\Rep|]=0$ or $\E[|\Flip|]=0$, we set $g_{W_R}=0$ or $g_{W_F}=0$, respectively.

Then for every $x,x' \in \obits^n$ and all $z\in \C$ such that $z$ is in the disks of convergence of all the above $g$ power series,
\begin{equation}\label{eq:rewritemeantracediff}
    \abs{\overline{P}_x(z)-\overline{P}_{x'}(z)}=\abs{P_x(g_M(z))-P_{x'}(g_M(z))} \cdot \abs{\E[|\Rep|]\cdot g_{W_R}(z)-\E[|\Flip|]\cdot g_{W_F}(z)}.
\end{equation}
\end{lem}

Now we wish to lower bound the two terms being multiplied on the right. Analogously to~\cite{DOS17,NP17}, we use the lemma below, due to Borwein and Erd\'elyi~\cite{BE97}, to lower bound\\
\begin{equation*}
    \abs{P_x(g_M(z))-P_{x'}(g_M(z))}.
\end{equation*}

\begin{lem}[\cite{BE97}]\label{lem:LBLittlewood}
There is a universal constant $c>0$ for which the following holds:
Let\\ ${\bf a}=(a_0,...,a_{\ell-1})\in\{-1,0,1\}^\ell$ be non-zero and define $A(w):=\sum_{j=0}^{\ell-1} a_j w^j$.
Let $\gamma_L$ denote any arc of the form \mbox{$\left\{e^{i\varphi}\ : \varphi\in [\theta, \theta+\frac{1}{L} ]\right\}$}.
Then, we have $\max_{w\in\gamma_L}|A(w)|\geq e^{-cL}$ for every $L>0$.
\end{lem}

This lemma implies that there is a constant \mbox{$c_1>0$} such that for every $L>0$ there exists $w_L=e^{i \varphi_L}$ with
$\abs{\varphi_L}\leq\frac{\pi}{L}$ satisfying
\begin{equation}
    \abs{P_x(w_L)-P_{x'}(w_L)}\geq e^{-c_1L}.\label{eq:BE}
\end{equation}
We can use~\eqref{eq:BE} to lower bound~\eqref{eq:rewritemeantracediff}, provided there exists $z_L$ such that $g_M(z_L)=w_L$ with good properties.
The following lemma ensures this.

\begin{lem}\label{lem:choicez}
For $L$ large enough there is a constant \mbox{$c>0$} such that for any $\varphi\in\left[-\frac{\pi}{L},\frac{\pi}{L}\right]$ there exists $z_\varphi$ satisfying \mbox{$g_M(z_\varphi)=e^{i\varphi}$}, \mbox{$1\leq \abs{z_\varphi}\leq 1+c\varphi^2$}. Moreover $|1-z_L|\geq c\varphi_L$.
\end{lem}
We prove Lemmas~\ref{lem:changevar} and \ref{lem:choicez} in Section~\ref{sec:proofs}.
\medskip

We bound the absolute value of the mean trace power series difference.
To ensure mean-based trace reconstruction is actually possible, we assume $\E[|\Rep|] \cdot g_{W_R}(z) \neq \E[|\Flip|] \cdot g_{W_F}(z)$. So, we split our proof into two cases: (i) $\E[|\Rep|]=\E[|\Flip|]$ and (ii) $\E[|\Rep|]\neq\E[|\Flip|]$.

\begin{enumerate}[(i)]

\item Assume that $\E[|\Rep|]=\E[|\Flip|]$. This implies $\E[|\Rep|]>0$, because if both were zero then we would have $\E[|\Rep|] \cdot g_{W_R}(z) = \E[|\Flip|] \cdot g_{W_F}(z)=0$ by our convention. Then
    \begin{align}\label{eq:mtpsdifference}
        \abs{\overline{P}_{x}(z) - \overline{P}_{x'}(z)}
        &=\abs{P_{x}(g_M(z)) - P_{x'}(g_M(z))}\cdot \E[|\Rep|]\cdot\abs{g_{W_R}(z)-g_{W_F}(z)}.
    \end{align}
    Here, $\E[|\Rep|]$ is a non-zero constant. We will lower bound the other components by
    \begin{equation}
        \abs{P_{x}(g_M(z)) - P_{x'}(g_M(z))} \geq e^{-cL} \text{ and }
        \abs{g_{W_R}(z)-g_{W_F}(z)}\geq\frac{1}{\poly(L)}
    \end{equation} for an appropriate $z=z_L$ and constant $c$.
    The polynomial $\frac{1}{2}\abs{P_{x}(g_M(z)) - P_{x'}(g_M(z))}$ is a Littlewood polynomial, so by Lemma~\ref{lem:LBLittlewood} and Lemma~\ref{lem:choicez} we can choose a $z_L$ such that $1\leq|z_L|\leq e^{c_2/L^2}$ and \mbox{$\abs{P_{x}(g_M(z)) - P_{x'}(g_M(z))} \geq e^{-c_3L}$} for some constants $c_2,c_3>0$.

    Now we show that $\abs{g_{W_R}(z_L)-g_{W_F}(z_L)}\geq\frac{1}{\poly(L)}$.
    By definition, the probability generating functions $g_{W_R}$ and $g_{W_F}$ can be written as power series:
    \begin{align*}
        g_{W_R}(z) &= \sum_{i=0}^\infty a_i (z_1)^i \\
        g_{W_F}(z) &= \sum_{i=0}^\infty b_i (z_1)^i,
    \end{align*}
    for some coefficients $a_i,b_i\in\R$.
    Let $d$ be the smallest index $i$ such that $a_i\neq b_i$. Then we can rewrite the power series difference as
    \begin{align*}
        g_{W_R}(z)-g_{W_F}(z)
            &= \sum_{i=0}^\infty (a_i-b_i) (z_1)^i \\
            &= (z-1)^d \sum_{j=0}^\infty (a_{j+d}-b_{j+d})(z-1)^d \\
            &= (z-1)^d q(z),
    \end{align*}
    where $q(z):=\sum_{j=0}^\infty (a_{j+d}-b_{j+d})(z-1)^d$.
    By definition, $q(z)$ approaches the non-zero constant $a_d-b_d$ when $z\to 1$. Hence, $g_{W_R}(z)-g_{W_F}(z) \sim (z-1)^d (a_d-b_d)$ when $z\to 1$.
    Therefore, for all $z$ sufficiently close to $1$ we have
    \begin{align}
        \abs{g_{W_R}(z)-g_{W_F}(z)}
            &\geq \frac{\abs{a_d-b_d}}{2}\cdot \abs{1-z}^d \nonumber\\
            &= \frac{\abs{a_d-b_d}}{2}\cdot \abs{1-z}^d \nonumber\\
            &= \poly(|1-z|). \label{eq:pgfDifferenceBound}
    \end{align}

    By Lemma~\ref{lem:LBLittlewood}, there exists a constant $c>0$ such that for any $L>0$ we have
    \begin{equation}
        \max_{\varphi\in[\frac{1}{2L},\frac{1}{L}]} \left\{\abs{f(e^{i\varphi})}\right\} \geq e^{-cL}.
    \end{equation}
    Let $\varphi_L$ denote the angle that achieves this maximum.
    Then by Lemma~\ref{lem:choicez},  there exists a $z_L$ such that $g_M(z_L)=e^{i\varphi_L}$ and $1\leq|z_L|\leq e^{c_2/L^2}$, and moreover $|1-z_L|\geq c_3\varphi_L$.

    Combining this with \eqref{eq:pgfDifferenceBound}, we obtain
    \begin{align*}
        \abs{g_{W_R}(z_L)-g_{W_F}(z_L)}
            &\geq \poly(|1-z_L|)\\
            &\geq \poly\left(\frac{c_4}{2L}\right)\\
            &= \frac{1}{\poly(L)}.
    \end{align*}
    Therefore, \eqref{eq:mtpsdifference} becomes
    \begin{align*}
        \abs{\overline{P}_{x-x'}(z)}
        &\geq e^{-c_3L}\cdot\frac{1}{\poly(L)}\\
        &\geq e^{-c_5L},
    \end{align*}
    for an appropriate constant $c_5$.

\item Assume that $\E[|\Rep|]\neq\E[|\Flip|]$. We can lower bound the mean trace polynomial obtained in Lemma 1 as follows.

    By the properties of probability generating functions, as $z$ approaches 1, the values $g_{W_R}(z)$ and $g_{W_F}(z)$ both approach 1.
    Then by assumption,
    \begin{align}
        \abs{\overline{P}_{x-x'}(z)}
        &=\abs{P_{x-x'}(g_M(z))}\cdot\abs{\E[|Rep|]\cdot g_{W_R}(z) -    \E[|Flip|]\cdot g_{W_F}(z)}\nonumber\\
        &\to\abs{P_{x-x'}(g_M(z))}\cdot c_6,
    \end{align}
    for some constant $c_6>0$ as $z\to 1$.

    Using Lemma~\ref{lem:choicez}, we showed in (i) that for large enough $L$ we can choose $z_L$ such that $1\leq|z_L|\leq e^{c_7/L^2}$ and $|P_{x-x'}(g_M(z_L))|\geq e^{-c_7L}$ for some constant $c_7$. This implies that for $z=z_L$ and large enough $L$,
    \begin{equation}
        \abs{\overline{P}_{x-x'}(z_L)}\geq e^{-c_8L}\cdot\frac{c_6}{2}.
    \end{equation}

\end{enumerate}

Taking $L=n^{1/3}$ for both cases (i) and (ii), we obtain the lower bound \mbox{$\abs{\overline{P}_x(z)-\overline{P}_{x'}(z)}\geq e^{-c' n^{1/3}}$} as claimed, for an appropriate constant $c'>0$.


\subsection{Proofs of Lemmas~\ref{lem:meantracediff} and~\ref{lem:choicez}}\label{sec:proofs}

In this section, we prove the remaining lemmas.

\begin{proof}[Proof of Lemma~\ref{lem:meantracediff}]
 Lemma~\ref{lem:meantracediff} is a corollary of Lemma~\ref{lem:changevar} below. The expression for the mean trace power series in Lemma~\ref{lem:changevar} immediately gives us the expression for the difference of two such power series in Lemma~\ref{lem:meantracediff}.

\begin{lem}\label{lem:changevar}
Suppose $\E[M]>0$ is finite. Let $W_R,W_F,W_+,W_-$ be distributions corresponding to $\Rep,\Flip,C_+$, and $C_-$ respectively, given by
\begin{equation*}
    W_R(j):=\frac{\Pr[j+1\in\Rep]}{\E[|\Rep|]} \text{ , }
    W_F(j):=\frac{\Pr[j+1\in\Flip]}{\E[|\Flip|]}, j=0,1,2,...,
\end{equation*}
and analogously for $C_+, C_-$. Also let $g_{W_R},g_{W_F},g_{W_+},g_{W_-}$, and $g_M$ be the probability generating functions corresponding to $W_R, W_F, C_+, C_-,$ and, $M$. If any of $\Rep,\Flip,C_+,C_-$ have expected size zero then this is not well defined, so set the corresponding $g$ to be the constantly zero function.

Let $\vec{1}$ be the length-$n$ string of all $1$s. Then for every $x\in\obits^n$ and $z\in\C$ such that $z$ is in the disks of convergence of all the above $g$ power series.
\begin{align*}
\overline{P}_x(z)
&= P_x(g_M(z))\left(\E[|\Rep|]\cdot g_{w_R}(z)-\E[|\Flip|]\cdot g_{w_F}(z) \right)\\
&\quad + P_{\vec{1}}(z) \prn{g_{W_+}(z)\cdot \E \brk{\abs{C_+}} - g_{W_-}(z)\cdot \E \brk{\abs{C_-}}}.
\end{align*}
\end{lem}

\begin{proof}[Proof of Lemma~\ref{lem:changevar}]
The channel acts on each input bit independently. For each input bit, the channel samples sets $\Rep,\Flip,C_+,$ and $C_-$ and produces an output string that depends on $x_i$ and those sets. Let $\Rep_i$ denote the set (and corresponding distribution) of coordinates $j\in[m]$ that are the result of replicating $x_i$. Similarly, $\Flip_i$ is the set of trace coordinates that result from the channel outputting $-x_i$ when acting on $x_i$. We also define $C_{-,i}$ and $C_{i,+}$ in this way.

If a trace coordinate $j$ is in $\Rep_i$, then the value at that coordinate is $x_i$. If a trace coordinate $j$ is in $\Flip_i$, then the value at that coordinate is $-x_i$. If a trace coordinate $j$ is in $C_{+,i}$, then the value at that coordinate is $+1$. If a trace coordinate $j$ is in $C_{-,i}$, then the value at that coordinate is $-1$. Combining all these observations, we can express the expected value at the trace's $j^\text{th}$ coordinate as
\begin{equation}
    \mu_{x,j} = \sum_{i=1}^n \Pr \brk{j \in \Rep_i}\cdot x_i - \Pr \brk{j \in \Flip_i}\cdot x_i + \Pr \brk{j \in C_{+,i}} - \Pr \brk{j \in C_{-,i}}.
\end{equation}
It follows immediately from the definition of the mean trace power series that

\begin{equation}\label{eq:meantracepowersum}
\overline{P}_x(z) = \sum_{j=1}^{\infty} \brk{\sum_{i=1}^n \Pr \brk{j \in \Rep_i}\cdot x_i - \Pr \brk{j \in \Flip_i}\cdot x_i + \Pr \brk{j \in C_{+,i}} - \Pr \brk{j \in C_{-,i}}}z^{j-1}.
\end{equation}

We show that
\begin{align*}
\sum_{j=1}^\infty\sum_{i=1}^n \Pr[j\in\Rep_i]\cdot x_i \cdot z^{j-1} &=\E[|\Rep|]\cdot g_{w_R}(z)\cdot P_x(g_M(z)) \\
\sum_{j=1}^\infty\sum_{i=1}^n \Pr[j\in\Flip_i]\cdot x_i \cdot z^{j-1} &=\E[|\Flip|]\cdot g_{w_F}(z)\cdot P_x(g_M(z)) \\
\sum_{j=1}^\infty\sum_{i=1}^n \Pr \brk{j \in C_{+,i}} \cdot z^{j-1} &= \E \brk{\abs{C_+}} \cdot g_{W_+}(z)\cdot P_{\vec{1}}(z)  \\
\sum_{j=1}^\infty\sum_{i=1}^n \Pr \brk{j \in C_{+,i}} \cdot z^{j-1}
&= \E \brk{\abs{C_-}} \cdot g_{W_-}(z)\cdot P_{\vec{1}}(z)
\end{align*}
where $\vec{1} \in \obits^n$ is the length-$n$ string of $1$'s. Plugging these equalities into~\eqref{eq:meantracepowersum} yields lemma~\ref{lem:changevar}.

Let \mbox{$M^{(\ell)}:=\sum_{k=1}^\ell M_k$}, where the $M_k:=|Y_{x_k}|$ denote the lengths of the channel outputs associated to each input bit $x_k$ and are i.i.d.\ according to $M$. For the first claim, we rewrite
\begin{equation}
    \sum_{j=1}^\infty\sum_{i=1}^n \Pr[j\in\Rep_i]\cdot x_i \cdot z^{j-1}=\sum_{i=1}^n x_i \sum_{j=1}^\infty \Pr[j\in\Rep_i] \cdot z^{j-1}.
\end{equation}
Then by the channel definition, we have
\begin{align}
    \sum_{j=1}^\infty \Pr[j\in\Rep_i]\cdot z^{j-1}&=\sum_{j=1}^\infty \Pr[M^{(i-1)}<j,j\in\Rep_i]\cdot z^{j-1}\nonumber
    \\
    &=\sum_{j=1}^\infty \Pr[j\in\Rep_i|M^{(i-1)}<j]\cdot \Pr[M^{(i-1)}<j]\cdot z^{j-1}\nonumber
    \\
    &=\sum_{j=1}^\infty\sum_{j'=0}^{j-1}\Pr[M^{(i-1)}=j']\cdot \Pr[j\in\Rep_i|M^{(i-1)}=j']\cdot z^{j-1}\nonumber
    \\
    &=\sum_{j=1}^\infty\sum_{j'=0}^{j-1}\Pr[M^{(i-1)}=j']\cdot  \Pr[j-j'\in \Rep]\cdot z^{j-1}\nonumber
    \\
    &=\sum_{j'=0}^\infty \Pr[M^{(i-1)}=j']\cdot \sum_{j=j'+1}^\infty Pr[j-j'\in \Rep]\cdot z^{j-1}\nonumber
    \\
    &=\sum_{j'=0}^\infty \Pr[M^{(i-1)}=j'] z^{j'}\cdot \sum_{j=1}^\infty \Pr[j\in \Rep]\cdot z^{j-1}\nonumber
    \\
    &=g_M(z)^{i-1}\cdot g_{W_R}(z)\cdot \E[|\Rep|].
\end{align}
We can interchange the sums above because $z$ is in the disk of convergence of $g_M$ and $g_{W_R}$.
The last step follows from the definition of $W_R$.
Hence,
\begin{equation}
    \sum_{i=1}^n x_i\cdot g_{W_R}(z)\cdot \E[|\Rep|]=P_x(g_M(z))\cdot g_{W_R}(z)\cdot \E[|\Rep|].
\end{equation}

\noindent
The remaining claims are proved in an identical way, by replacing $\Rep$ and $g_{W_R}$ with the appropriate distribution and power series.
\end{proof}

\end{proof}

We prove Lemma~\ref{lem:choicez} using the standard inverse function theorem stated below.
\begin{lem}[\protect{\cite[Section VIII.4]{Gam03}, adapted}]\label{lem:invfunc}
Let $g:\Omega\to\C$ be a non-constant function analytic on a connected open set $\Omega\subseteq \C$ such that $g'(z)\neq 0$ for a given $z\in\Omega$.
Then, there exist radii $\rho,\eps>0$
such that for every $w\in \cD_\eps(g(z))$ there exists a unique $z_w\in \cD_\rho(z)$ satisfying $g(z_w)=w$.
Moreover, the inverse function $f:\cD_\eps(g(z))\to \cD_\rho(z)$ defined as \mbox{$f(w)=z_w$} is analytic on $\cD_\eps(g(z))$.
\end{lem}

\begin{proof}[Proof of Lemma~\ref{lem:choicez}]
Because $M$ is sub-exponential and not always zero,
$g_M$ is a non-constant analytic function on some open ball $\cD_r(0)$ of radius $r>1$ that satisfies \mbox{$g'_M(1)=\E[M]\neq 0$}.
Hence, Lemma~\ref{lem:invfunc} applies with \mbox{$g=g_M$}, so there exist
$\rho,\eps>0$ and an analytic function \mbox{$f:\cD_\eps(1)\to\cD_\rho(1)$} such that $g_M(f(w))=w$.
In particular, there exists $\gamma\in(0,\eps)$ such that for every $w\in\cD_\gamma(1)$ we can write
\begin{equation}\label{eq:Taylorexp}
    f(w)=1+f'(1)(w-1)+\sum_{i=2}^\infty \frac{f^{(i)}(1)}{i!}(w-1)^i.
\end{equation}
This is because $f(1)=1$, since $g(1)=1$, and furthermore
\begin{equation*}
    \sum_{i=2}^\infty \abs{\frac{f^{(i)}(1)}{i!}}\cdot \abs{w-1}^i\leq c''\abs{w-1}^2
\end{equation*}
for some constant $c''>0$ and all such $w$.
Assume that $L$ is large enough so that $e^{i\varphi}\in\cD_\gamma(1)$ for all $\varphi\in\left[-\frac{\pi}{L},\frac{\pi}{L}\right]$.
Then, we set $z_\varphi=f(e^{i\varphi})$.
Note that \mbox{$g_M(z_\varphi)=e^{i\varphi}$} by the definition of $f$, as required.
Combining~\eqref{eq:Taylorexp} with $w=e^{i\varphi}$ and the triangle inequality, we have
\begin{equation*}
    \abs{z_\varphi-1}=O\left( \left|e^{i\varphi}-1\right| \right)\to 0
\end{equation*}
as $L\to\infty$.
Since $g_W$ is a continuous function on a neighborhood of $1$, and $g_W(1)=1$, it follows that $\abs{g_W(z_\varphi)}\geq 1/2$ if $L$ is large enough.
On the other hand, combining~\eqref{eq:Taylorexp} with the fact that
\begin{equation*}
    f'(1)=\frac{1}{g'(f(1))}=\frac{1}{g'(1)}=\frac{1}{\E[M]}\in\R,
\end{equation*}
by the chain rule, we obtain
\begin{align*}
    \abs{z_\varphi}&\leq \abs{1+\frac{e^{i\varphi}-1}{\E[M]}}+c''\abs{e^{i\varphi}-1}^2\\
    &\leq\sqrt{\left(1-\frac{1-\cos(\varphi)}{\E[M]}\right)^2+\frac{\sin(\varphi)^2}{\E[M]^2}}+c''\varphi^2\\
    &\leq \sqrt{1+\frac{2(1-\cos(\varphi))}{\E[M]^2}}+c''\varphi^2\\
    &\leq 1+\left(\frac{1}{\E[M]^2}+c''\right)\varphi^2.
\end{align*}
The second inequality holds because $\abs{e^{i\varphi}-1}\leq \varphi$.
The last inequality follows by noting that $1-\cos(\varphi)\leq\varphi^2/2$ and $\sqrt{1+x}\leq 1+x$ for $x\geq 0$.
Finally, we prove the last inequality in the statement.
By definition,
\begin{equation}
    z_L=f(e^{i\varphi_L})=1+f'(1)(e^{i\varphi_L}-1)+\sum_{j=2}^\infty\frac{f^{(j)}(1)}{j!}(e^{i\varphi_L}-1)^j.
\end{equation}
Here  $f'(1)\neq 0$ for a sufficiently large $L$. Then since $|e^{i\varphi_L}-1|\geq -|\varphi_L|$, we can write
    \begin{align*}
        |z_L-1|
            &= \abs{f'(1)(e^{i\varphi_L}-1)+\sum_{j=2}^\infty
                \frac{f^{(j)}(1)}{j!}(e^{i\varphi_L}-1)^j} \\
            &\geq \abs{f'(1)}\abs{e^{i\varphi_L}-1}+
                \abs{\sum_{j=2}^\infty\frac{f^{(j)}(1)}{j!}(e^{i\varphi_L}-1)^j} \\
            &= \Theta(\varphi_L)-\Theta(\varphi_L^2)\\
            &= \Theta(\varphi_L).
    \end{align*}
    Using this with the fact that $\varphi_L\geq \frac{1}{2L}$ by definition, we obtain $|1-z_L| = |z_L-1|\geq \frac{c}{2L}$ for some constant $c$.
\end{proof}

\section{Proof of Lemma~\ref{lem:tail}}\label{sec:prooftail}
To conclude the argument, we prove Lemma~\ref{lem:tail} using an argument analogous to~\cite[Appendix A.2]{DOS17} and the fact that $M$ is sub-exponential.

Let $M_1,M_2,\dots,M_n$ be i.i.d.\ according to $M$, and set $M^{(n)}=\sum_{i=1}^n M_i$.
Then, we have
\begin{equation*}
    \abs{\mu_{x,i}-\mu_{x',i}}\leq \Pr[M^{(n)}\geq i]
\end{equation*}
for every $i$.
Since $M$ is sub-exponential, a direct application of Bernstein's inequality~\cite[Theorem 2.8.1]{Ver18} guarantees the existence of constants $c_4,c_6>0$ such that for $N=c_4 n$ and any $j\geq 1$ we have
\begin{equation*}
    \Pr[M^{(n)}\geq N+j]\leq 2e^{-c_6(N+j)}.
\end{equation*}
Combining these observations with the assumption that \mbox{$|z|\leq e^{c_3 n^{-2/3}}$} yields
\begin{align*}
        \sum_{i=N+1}^\infty \abs{\mu_{x,i}-\mu_{x',i}}|z|^{i-1}&\leq \sum_{j=1}^\infty 2 e^{-c_6(N+j)}\cdot e^{c_3 n^{-2/3}} \\
        &\leq e^{-c_5 n}
\end{align*}
for some constant $c_5>0$ and $n$ large enough.

\section{The Conditions in Theorem~\ref{thm:mainresult} are Necessary}

 Our main result applies to oblivious synchronization channels that meet two conditions:  $M$ must be a sub-exponential random variable and distributions formed from $\Rep$ and $\Flip$ must not satisfy a certain equality. It is natural to ask whether these conditions are necessary. We do not know whether the sub-exponentiality is necessary, but we now demonstrate that the requirements on $\Rep$ and $\Flip$ are.

\begin{coro}\label{coro:asumptionneeded}
Let $\channel$ be an oblivious synchronization channel where $M$ is a sub-exponential random variable. If $\E[|\Rep|] \cdot g_{W_R}(z) = \E[|\Flip|] \cdot g_{W_F}(z)$, then mean-based trace reconstruction is impossible.
\end{coro}
\begin{proof}
Let $x,x' \in \obits^n$. We show that their mean traces are identical. By combining our assumption that $\E[|\Rep|] \cdot g_{W_R}(z) = \E[|\Flip|] \cdot g_{W_F}(z)$ with corollary~\ref{lem:meantracediff}, we have that
\begin{align*}
\abs{\overline{P}_x(z)-\overline{P}_{x'}(z)} &=\abs{P_x(g_M(z))-P_{x'}(g_M(z))} \cdot \abs{\E[|\Rep|]\cdot g_{w_R}(z)-\E[|\Flip|]\cdot g_{w_F}(z)} \\
&= 0
\end{align*}
which implies that the power series $\overline{P}_x(z)-\overline{P}_{x'}(z)$ has all zero coefficients. This is the same as saying that for all $i$, $\mu_{x,i} = \mu_{x',i}$. This holds for any pair of messages $x,x'$, so all possible messages result in the same mean trace. Thus, mean-based trace reconstruction is impossible.
\end{proof}

This corollary by itself does not show that we have to state the assumption $\E[|\Rep|] \cdot g_{W_R}(z) \neq \E[|\Flip|] \cdot g_{W_F}(z)$ in our theorem.
If all the channels we are concerned with had this property, then the theorem's assumption would be redundant and the above corollary would be vacuously true.

It turns out that there are channels where this assumption does not hold. We prove that such channels exist by describing one. Specifically, we describe an oblivious synchronization channel where $\E[|\Rep|] \cdot g_{W_R}(z) = \E[|\Flip|] \cdot g_{W_F}(z)$ and mean-based trace reconstruction is impossible, but (non-mean-based) trace reconstruction is easy.

\subsection{A Channel Where Mean-based Trace Reconstruction is Impossible}
Let $M=2$ be a constant, i.e. for each input bit the channel always outputs two bits. $C_- = \emptyset$ always. $\textsf{Rep}, \textsf{Flip},$ and $C_+$ are jointly distributed among three equally likely outcomes.

\begin{center}
\begin{tabular}{llll}
 & \textsf{Rep} & \textsf{Flip}  & $C_+$  \\ \hline
\multicolumn{1}{l|}{Output 1} & $\left\{1,2  \right\}$ & $\emptyset$  & $\emptyset$  \\
 \multicolumn{1}{l|}{Output 2} & $\emptyset$ & $\left\{ 1  \right\}$  & $\left\{ 1  \right\}$  \\
 \multicolumn{1}{l|}{Output 3} & $\emptyset$ & $\left\{ 2  \right\}$  & $\left\{ 1  \right\}$
\end{tabular}
\end{center}

Equivalently, for a single input bit, the channel outputs the following three strings with equal probability.
\begin{center}
\begin{tabular}{llll}
\multicolumn{1}{l|}{Input} & $-1$ &  $+1$ &  \\ \hline
\multicolumn{1}{l|}{Output 1} & $(-1, -1)$    & $(+1,+1)$     &  \\
\multicolumn{1}{l|}{Output 2} & $(+1,+1)$     & $(-1,+1)$     &  \\
\multicolumn{1}{l|}{Output 3} & $(+1,+1)$     & $(+1,-1)$     &
\end{tabular}
\end{center}

Regardless of input bit, both output bits have expected value $1/3$. Thus, mean-based trace reconstruction is impossible.

However, trace reconstruction is easy. On input $-1$, the channel outputs $-1,-1$ with probability $1/3$. On input $+1$, the channel never outputs $-1,-1$. This can be used to distinguish any pair of strings. Let $x,y \in \obits^n$ be two distinct strings.
Without loss of generality, suppose $x_i = -1$ and $y_i = +1$ for an arbitrary index $i$. Let $\textsf{Tr}(x)$ and $\textsf{Tr}(y)$ denote traces of $x$ and $y$. Then
\[ \Pr\left[ \textsf{Tr}(x)_{2i-1} = \textsf{Tr}(x)_{2i}  = -1 \right] = 1/3\]
 and
\[ \Pr\left[ \textsf{Tr}(y)_{2i-1} = \textsf{Tr}(y)_{2i}  = -1 \right] = 0.\] This leads to a simple reconstruction algorithm, which outputs a string \mbox{$z \in \obits^n$} after looking at the traces of a string $z_{\text{true}} \in \obits^n$.
\begin{enumerate}
\item
Look at $t$ independent traces of $z_{\text{true}}$.

\item
If any of the traces have $(-1,-1)$ in the $(2i - 1)^\text{th}$ and $(2i)^\text{th}$ coordinates, set $z_{i} = -1$. Otherwise, set $z_i = +1$.

\item
Output $z$.
\end{enumerate}
To analyze this algorithm's correctness, we consider the probability that a fixed coordinate is correct, then union bound over all coordinates. Let $j \in [n]$ be a fixed coordinate. If $(z_{\text{true}})_{j} = +1$, then $z_{j}=+1$ with probability $1$. If $(z_{\text{true}})_{j} = -1$, then $z_{j} = +1$ with probability $\prn{2/3}^t$. Combining this with a union bound over the $n$ coordinates, we get
\begin{align*}
\Pr \brk{z_{\text{guess}} \neq z_{\text{true}}} &\leq n \cdot \Pr \brk{\prn{z_{\text{guess}}}_j \neq \prn{z_{\text{true}}}_{j}} 
\\
&\leq n \cdot \prn{2/3}^t.
\end{align*}
Hence,  $O(\log(n/\delta))$ traces are enough to reconstruct the string with probability at least $1-\delta$.

\section{Future Work}

We have shown that $\exp(O(n^{1/3}))$ traces suffice for mean-based worst-case trace reconstruction over a broad class of oblivious synchronization channels.
Because $\exp(\Omega(n^{1/3}))$ traces are required for mean-based worst-case trace reconstuction over the deletion channel, this means that our result cannot be improved in general.
However, our channel model does not cover all discrete memoryless synchronization channels as defined by Dobrushin~\cite{Dob67,CR20}.
It would be interesting to extend our result in some form to all such non-trivial channels.
On the other hand, to complement the above, it would be interesting to prove trace complexity lower bounds for mean-based reconstruction over all these channels.
Furthermore, it is unclear whether the assumption that $M$ is sub-exponential is necessary for our result.
A clear extension of this work would be to either remove this condition or prove that it is necessary for mean-based trace reconstruction from $\exp(O(n^{1/3}))$ traces.

\bibliographystyle{IEEEtran}
\bibliography{refs}

\end{document}